\documentclass{entcs} \usepackage{prentcsmacro}
\input pdfcolor.tex 
\def\lastname{D\'{\i}az-Caro, Arrighi, Gadella and Grattage}

\usepackage{amsmath}
\usepackage[inference]{semantic}

\usepackage[matrix,frame,arrow]{xy}
\newcommand{\ket}[1]{\left\vert{#1}\right\rangle}
\newcommand{\qw}[1][-1]{\ar @{-} [0,#1]}
\newcommand{\qwx}[1][-1]{\ar @{-} [#1,0]}
\newcommand{\cw}[1][-1]{\ar @{=} [0,#1]}
\newcommand{\cwx}[1][-1]{\ar @{=} [#1,0]}
\newcommand{\gate}[1]{*{\xy *+<.6em>{#1};p\save+LU;+RU **\dir{-}\restore\save+RU;+RD **\dir{-}\restore\save+RD;+LD **\dir{-}\restore\POS+LD;+LU **\dir{-}\endxy} \qw}
\newcommand{\meter}{\gate{\xy *!<0em,1.1em>h\cir<1.1em>{ur_dr},!U-<0em,.4em>;p+<.5em,.9em> **h\dir{-} \POS <-.6em,.4em> *{},<.6em,-.4em> *{} \endxy}}
\newcommand{\control}{*-=-{\bullet}}
\newcommand{\controlo}{*!<0em,.04em>-<.07em,.11em>{\xy *=<.45em>[o][F]{}\endxy}}
\newcommand{\ctrl}[1]{\control \qwx[#1] \qw}
\newcommand{\targ}{*{\xy{<0em,0em>*{} \ar @{ - } +<.4em,0em> \ar @{ - } -<.4em,0em> \ar @{ - } +<0em,.4em> \ar @{ - } -<0em,.4em>},*+<.8em>\frm{o}\endxy} \qw}
\newcommand{\gategroup}[6]{\POS"#1,#2"."#3,#2"."#1,#4"."#3,#4"!C*+<#5>\frm{#6}}
\newcommand{\rstick}[1]{*!L!<-.5em,0em>=<0em>{#1}}
\newcommand{\lstick}[1]{*!R!<.5em,0em>=<0em>{#1}}
\newcommand{\Qcircuit}{\xymatrix @*=<0em>}
\newcommand{\clet}{\mathbf{let}\ }
\newcommand{\cin}{\ \mathbf{in}\ }
\newcommand{\cif}{\ \mathbf{if}\ }
\newcommand{\cthen}{\ \mathbf{then}\ }
\newcommand{\rif}[1]{\textrm{IF-}\ket{#1}}
\newcommand{\celse}{\ \mathbf{else}\ }
\newcommand{\f}[1]{\textrm{\sf{#1}}}
\newcommand{\alice}{\f{alice}}
\newcommand{\bob}{\f{bob}}
\newcommand{\epr}{\f{epr}}
\newcommand{\teleport}{\f{teleport}}
\newcommand{\cnot}{\textit{cnot}\,}

\newcommand {\aplica}[1]{\f{#1}}
\newcommand{\flecha}[1]{\rightarrow_{#1}}
\newcommand{\tte}[1]{\stackrel{#1}\twoheadrightarrow}
\newcommand{\tps}[2]{\stackrel{#1}\rightarrow_{#2}}
\newcommand{\te}[2]{\{\langle{#1},\, {#2}\rangle\}}
\newcommand{\meas}{\textrm{M}}
\newcommand{\id}{\textrm{Id}}
\DeclareMathAlphabet{\letralinda}{OT1}{pzc}{m}{it}
\newcommand{\hist}{\letralinda{H}}

\begin{document}
\begin{frontmatter}
  \title{Measurements and confluence in quantum lambda calculi with explicit qubits}
  \author[UNR]{Alejandro D\'{\i}az-Caro\thanksref{adc-email}},
  \author[LIG]{Pablo Arrighi\thanksref{pa-email}},
  \author[UVa]{Manuel Gadella\thanksref{mg-email}} and
  \author[LIG]{Jonathan Grattage\thanksref{jg-email}}
  \address[UNR]{Departamento de Ciencias de la Computaci\'on, Universidad Nacional de Rosario, Argentina}
  \address[LIG]{Laboratoire d'Informatique de Grenoble, Universit\'e de Grenoble, France}
  \address[UVa]{Departamento de F\'{\i}sica Te\'orica, At\'omica y \'Optica, Universidad de Valladolid, Spain}

\thanks[adc-email]{Email: \href{mailto:alejandro.diaz-caro@imag.fr} {\texttt{\normalshape alejandro.diaz-caro@imag.fr}}}
\thanks[pa-email]{Email: \href{mailto:pablo.arrighi@imag.fr} {\texttt{\normalshape pablo.arrighi@imag.fr}}}
\thanks[mg-email]{Email: \href{mailto:gadella@fta.uva.es} {\texttt{\normalshape gadella@fta.uva.es}}}
\thanks[jg-email]{Email: \href{mailto:jonathan.grattage@ens-lyon.fr}{\texttt{\normalshape jonathan.grattage@ens-lyon.fr}}}

\begin{abstract} 
This paper demonstrates how to add a measurement operator to quantum
$\lambda$-calculi. A proof of the consistency of the semantics is given through
a proof of confluence presented in a sufficiently general way to allow this
technique to be used for other languages. The method described here may be
applied to probabilistic rewrite systems in general, and to add measurement to
more complex languages such as {\em QML} \cite{jon} or {\em Lineal}
\cite{arrighi3}\cite{arrighi1}, which is the subject of further research. 
\end{abstract}
\begin{keyword}
    Quantum lambda calculus, Measurement, Confluence, Probabilistic rewrite
system
\end{keyword}
\end{frontmatter}

\section{Introduction}
In the quest to develop quantum programming languages, quantum extensions of
functional languages provide a promising route, hence the explosion of works on
quantum lambda calculi and quantum functional languages
\cite{arrighi1}\cite{jon}\cite{lambdaclassical}\cite{qlambda}. 
The current language proposals can be split into two categories. In the first
category, qubits are manipulated as pointers towards a quantum memory
\cite{prost}\cite{lambdaclassical}, thus the syntax does not provide an explicit
description of the qubits. It does, however, together with a linear type system,
give a convenient and coherent way to handle operations on qubits. A drawback is
that the semantics of quantum operations cannot be given intrinsically in the
syntax, as this would require the actual state of the quantum memory to be
known. In the second category of languages
\cite{arrighi1}\cite{jon}\cite{qlambda} the description of the qubits is part of
the programming language, and no type system is required. An advantage here is
that the entire semantics can be expressed simply as a rewrite system between
terms of the language. This turns into a weakness regarding measurements,
because the inherently probabilistic nature of measurement makes it difficult to
express as part of a rewrite system. In fact, neither category of languages
allow this feature. \cite{arrighi1}\cite{qlambda}

The case of Altenkirch and Grattage's {\em QML} \cite{jon} is not so clear-cut,
but it does illustrate this 
difficulty. \emph{QML} includes
measurements with an operational semantics given in terms of quantum circuits.
However, the corresponding algebraic theory \cite{vizzotto} stands only for a
pure quantum subset of the language, with classical-control and measurement
omitted.

Van Tonder's $\lambda_q$ \cite{qlambda} is a higher-order untyped lambda
calculus which includes quantum properties. This calculus carries a history
track to keep the necessary information to invert reductions, to ensure that the
global computation process is unitary. It is closely related to linear logic,
with the syntax being a fragment of the one introduced by Wadler
\cite{Wadler94}, extended with constants to represent quantum entities such as
qubits and gates. Linearity concepts are used to distinguish definite terms from
arbitrary superposition terms.  These syntactic markers constitute the main
difference with Arrighi and Dowek's {\em Lineal} \cite{arrighi3}\cite{arrighi1},
which is more permissive. As mentioned
previously, measurement is not included in these two proposals.

The work presented here shows how to add measurement to a quantum lambda
calculus with explicit qubits in an elegant manner. This is done with full
details for the $\lambda_q$-calculus, with a proof that confluence, and hence
the consistency of the operational semantics, is preserved by this extension.
Although this calculus does not need a proof of confluence in the original setting, due to the fixed reduction strategy, this proof is necessary in the presence of measurement.  Furthermore,  it is non-trivial and  has the novelty of showing the confluence in a
probabilistic setting with the branching produced by the measurement.
The methods illustrated here are general, and applying these techniques to {\em
QML} and {\em Lineal} is in progress.
\bigskip

% --Motivations
In contrast to measurement in classical mechanics, which gives the value of a
given observable with an associated error, measurements in quantum mechanics
have an intrinsically probabilistic character. That is, a quantum measurement
can give, \emph{a priori}, a certain number of results, each one with some
finite probability. Moreover, the state of the system after the measurement is
changed in an irreversible manner by the act of measurement. This unintuitive
behaviour is of acute importance in quantum information processing.

Measurement is a key property in many quantum information processing tasks, such
as quantum cryptography, superdense coding, and in quantum search algorithms.
Not having measurements can lead to misinterpretations. Consider as an example
the quantum teleportation algorithm with deferred measurement \cite{qlambda} as
defined in Fig. \ref{alg:telep}. Here it is unclear if Alice and Bob can be
physically separated, as all the channels used are quantum channels. An obvious
question arises: why use this algorithm if there is a quantum channel between
Alice and Bob? Measuring the final state will result in the original
logical-qubit having been transferred to Bob. The problem is not one of
correctness, but of interpretation.

Secondly, understanding measurement is essential to avoid misinterpreting
quantum computation as a whole (e.g. why quantum computation does not lead
straightforwardly to an exponential jump in complexity). This work takes the
view that in order to understand the possibilities and limitations of quantum
computation, measurement needs be formalised in an elegant manner. Note that the
projective measurement discussed in this paper is not the only possibility for a
quantum measurement, but it is one of the simplest. In addition, any quantum
measurement can be reproduced by the action of a unitary mapping and a
projective measurement. 

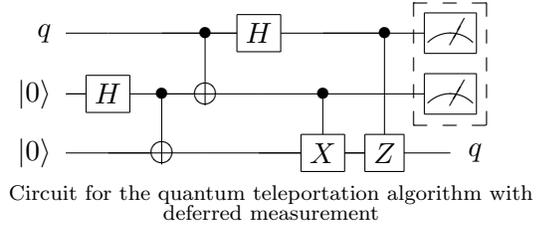
\begin{figure}[htp]
 \centering
\begin{minipage}{6.5cm}
{\small\medskip
\noindent$\teleport\ q \rightarrow
   	\begin{aligned}[t]
		&\clet (e_1, e_2) = \epr \cin \\
		&\clet (q', y') = \alice\ (q, e_1) \cin \\
		&\quad \bob\ (q', y', e_2)
	\end{aligned}$

where

$\begin{array}{l}
\alice\ (q, e_1) \rightarrow
     	\begin{aligned}[t]
       			&\clet (q', y') = \cnot\ (q, e_1) \cin \\
			& ((H\ q'), y')
     	\end{aligned}\\
\bob\ (q', y', e_2) \rightarrow
    	\begin{aligned}[t]
      			&\clet (y'', e_2') = \textit{cX}\,\ (y', e_2) \cin \\
      			&\clet (q'', e_2'') = \textit{cZ}\,\ (q', e_2') \cin \\
      			&\quad (q'', y'', e_2'')
    	\end{aligned}\\
\epr \equiv \cnot\ ((H\ 0),0)\end{array}$}
\medskip
\end{minipage}
\   \
\hfill \begin{minipage}{7cm}
\begin{eqnarray*}
\Qcircuit @C=.7em @R=.7em {
\lstick{q} & \qw 				& \qw	 	& \ctrl{1} \qw 
& \gate{H} \qw & \qw 			& \ctrl{2} \qw 	&  \meter  \\
\lstick{\ket{0}}	& \gate{H} \qw		& \ctrl{1}	& \targ \qw    
& \qw          & \ctrl{1} 		& \qw 			&  \meter  \\
\lstick{\ket{0}} 	&  \qw 				& \targ \qw	& \qw   
      	& \qw          & \gate{X} \qw  	& \gate{Z} \qw 	& \qw &  \lstick{q}
\gategroup{1}{8}{2}{8}{.7em}{--}
}
\end{eqnarray*}
{\begin{center}\small Circuit for the quantum teleportation algorithm with
deferred measurement\end{center}}
\end{minipage}
 \caption{Teleportation algorithm in non-extended $\lambda_q$.}
 \label{alg:telep}
\end{figure}
\bigskip

% --Structure of the paper
In the second section of this paper, the process of adding measurement is shown
with full details for van Tonder's $\lambda_q$. The section concludes with an
implementation of the teleportation algorithm in extended $\lambda_q$. Section 3
discusses and proves confluence for extended $\lambda_q$. Finally, section 4
closes with details of ongoing and future work.

\section{Adding measurement}

Adding a measurement operator to a quantum lambda calculus can be achieved with
only small changes to the grammar. In this section we show how to change the
syntax, add well-formedness rules for terms, and give the operational semantics.

\subsection{Syntax}
To account for measurements, the grammar of $\lambda_q$ must be extended with a
family of measurement operators $M_I$, which measure the qubits indicated by the
set $I$. In addition, it is necessary to make the syntax for qubits precise,
because their ``shape'' is needed by the measurement operator. This is achieved
in a manner following on from {\em Lineal} \cite{arrighi1} and \emph{QML}
\cite{jon}. Regarding van Tonder's original syntax, the only significant change
is to split ``constants'' into qubit-constants, measurement-constants and
gate-constants. The extended syntax is shown in Figure \ref{fig:syntax} and the
added rules of well-formedness are given in Figure \ref{fig:well-formed}.

A term is a pre-term produced by the syntax in Figure \ref{fig:syntax} which
follows the rules for well-formedness given by van Tonder \cite{qlambda} plus
the rules in Figure \ref{fig:well-formed}. Amongst these rules note that
\textit{M} and \textit{Gate} state that $M_I$ and $c_U$ are simply constant
symbols. \textit{Zero} and \textit{One} force $\ket{0}$ and $\ket{1}$
respectively to be non-linear terms. \textit{Tensor} and \textit{!Tensor} allow
tensorial products between qubits to be written.
Although terms like $(c_U\ q)\otimes q$ are not allowed, they are a contraction
for $c_{U\otimes I}\ (q\otimes q)$.
\textit{Superposition} provides a way of writing qubits in superpositions, and
\textit{Simplification} allows subterms with the scalar factor $0$ to be
removed.

Note that a term with a pattern $!q\otimes q$ is not well-formed, but there is
always an equivalent term which can express this in a well-formed way. For
example, the term $!\ket{0}\otimes(\alpha!\ket{0}+\beta!\ket{1})$ is not
well-formed, however, it is equivalent to
$\alpha(!\ket{0}\otimes!\ket{0})+\beta(!\ket{0}\otimes!\ket{1})$ which is
well-formed.

\begin{figure}[!htbp]
\centering
$$\begin{array}{rll}
t ::=	&			& \emph{Pre-terms:} \\
	& x			& \qquad \emph{variable} \\
	& (\lambda x.t)	& \qquad \emph{abstraction} \\
	& (t\ t)		& \qquad \emph{application} \\
	& !t			& \qquad \emph{nonlinear term} \\
	& (\lambda !x.t)	& \qquad \emph{nonlinear abstraction} \\
	& c_U			& \qquad \emph{gate-constant} \\
	& q			& \qquad \emph{qubit-constant} \\
	& M_I			& \qquad \emph{measurement-constant} \\
q::= 	&			& \emph{Qubit-constants:} \\
	& \ket{0}\;|\;\ket{1}	& \qquad \emph{base qubit} \\
	& (q \otimes q)	& \qquad \emph{tensorial product} \\
	& (q + q)		& \qquad \emph{superposition} \\
	& \alpha (q)		& \qquad \emph{scalar product} \\
c_U::=	&			& \emph{Gate-constants:} \\
	& H\ |\ \cnot\ |\ X\ |\ Z\ |\ \dots  &
\end{array}$$
\caption{Syntax for extended $\lambda_q$.}
\label{fig:syntax}
\end{figure}

\begin{figure}[!htbp]
 \centering
$$\begin{array}{c}
\inference{I \subset \mathbb{N}}{\vdash M_I}[M]\hspace{4cm}
\inference{}{\vdash c_U}[Gate]\\
\\
\inference{}{\vdash !\ket{0}}[Zero]\hspace{4cm}
\inference{}{\vdash !\ket{1}}[One]\\
\\
\inference{\Gamma \vdash q_1 \qquad \Delta \vdash q_2}{\Gamma, \Delta \vdash q_1
\otimes q_2}[Tensor]\hspace{1cm}
\inference{\Gamma \vdash !q_1 \qquad \Delta \vdash !q_2}{\Gamma, \Delta \vdash
!q_1 \otimes !q_2}[!Tensor]\\
\\
\inference{\sum\limits_{i=0}^{2^n-1}{|\alpha_i|^2}=1 \qquad \alpha_i \in
\mathbb{C}, i=0\dots 2^n-1}{\vdash \alpha_0 (!\ket{0}\otimes \cdots \otimes
!\ket{0}) + \cdots + \alpha_{2^n-1} (!\ket{1}\otimes \cdots \otimes
!\ket{1})}[Superposition]\\
\\
\inference{\alpha_r=0, r \in \{0,\dots,2^n-1\} \qquad \Gamma \vdash
\sum\limits_{i=0}^{2^n-1}{\alpha_i q_i}}{\Gamma \vdash
\sum\limits_{\begin{subarray}{l} i=0 \\ i\neq r \end{subarray}}^{2^n-1}{\alpha_i
q_i}}[Simplification]
\end{array}$$
\caption{Rules for well-formedness added to $\lambda_q$.}
\label{fig:well-formed}
\end{figure}

\begin{note}
The usual $\clet$construction will be used as a useful shorthand, defined as:
\[
\begin{array}{c}
\clet x=t \cin u\textrm{ gives }(\lambda x.u)\ t \\
\clet !x=t \cin u\textrm{ gives }(\lambda!x.u)\ t
\end{array}
\]
It is interesting to note that a cloning machine such as $\lambda x.(\clet y=x
\cin y\otimes y)$ is syntactic-sugar for $\lambda x.((\lambda y. y \otimes y)\
x)$, which is forbidden by the well-formedness rules since $y$ is linear (it
cannot appear twice), and moreover there is no way to tensor variables: they can
only be qubit-constants.

$\clet$can also be used over lists, as per van Tonder's $(x, y)$, but they are
written here as a tensor product. For example, the term
$\clet x \otimes y = M_{\{1,2\}}\ (q_1\otimes q_2) \cin t$
is the same as
$\clet (x,y) = (M_{\{1\}}\ q_1, M_{\{1\}}\ q_2) \cin t$.
Additionally, note that $x\otimes y$ is used following van Tonder's $(x, y)$; it
is an overloading of the operator $\otimes$, denoting both the tensor product
between qubits and also list constructors.
\end{note}

\subsection{Operational Semantics}
Measurement in quantum systems is an inherently probabilistic operation.
Following Di Pierro \textit{et al.} \cite{prolambda}, where a probabilistic
rewrite system is defined over a $\lambda$-calculus, the operational semantics
for measurement in extended $\lambda_q$ is defined as follows:
$$\begin{array}{c}
\inference[M: ]{q=\sum\limits_{u=0}^{2^m-1}{\alpha_u q^{(u)}}}
		{\hist;(M_I\  q)\flecha{p_w} \sum\limits_{u \in
C(w,m,I)}{\frac{\alpha_u}{\sqrt{p_w}} q^{(u)}}}[$\forall i\in I, 1\leq i\leq m$]
\end{array}$$

\noindent where
\begin{itemize}
	\item $w=0,\dots,2^{|I|}-1$.
	\item $q^{(u)} =
!q_{1}^{(u)}\otimes!q_{2}^{(u)}\otimes\dots\otimes!q_{m}^{(u)}$ with
$!q_{k}^{(u)}=!\ket{0} \textrm{ or } !\ket{1}$ for $k=1\dots m$.
	\item $C(w,m,I)$ is the set of binary words of length $m$ such that they
coincide with $w$ on the letters of index $I$.
	\item $p_w = \sum\limits_{u \in C(w,m,I)}{|\alpha_u|^2}$. 
	\item The notation $t\flecha{p}t'$ means that $t$ goes to $t'$ with
probability $p$.
\end{itemize}
\medskip

It is instructive to look at an example of this rule in action:

\begin{example}
Let $m=5$, $I=\{2,3,5\}$ and $q=\sum\limits_{u=0}^{2^5-1}{\alpha_u
(\bigotimes\limits_{k=1}^5 !q_{k}^{(u)})}$ with $!q_{k}^{(u)}=!\ket{x}$ and $x$
is the $k^{th}$ bit in the binary representation of $u$. According to the
previous rule, $(M_I\ q)$ will generate $2^{|I|}=8$ different outputs
(corresponding to the different possible values of the qubits $2$, $3$ and $5$,
which are measured). Take as an example the output $w=2$ (its 3-bit binary
representation is $010$). Hence, $C(2,5,I) = \{4, 6, 20, 22\}$ which are the
numbers $u$ between $0$ and $2^5-1$ whose binary representation is of the form
$x01y0$ (so they coincide with $w$, if we compare the bits $2$, $3$ and $5$ of
$u$ with the bits $1$, $2$ and $3$ of $w$). Then, the final term is:
$$\frac{\alpha_{4}}{\sqrt{p_2}} q^{(4)} + \frac{\alpha_{6}}{\sqrt{p_2}} q^{(6)}
+ \frac{\alpha_{20}}{\sqrt{p_2}} q^{(20)} + \frac{\alpha_{22}}{\sqrt{p_2}}
q^{(22)}$$
\noindent where \\
$\begin{array}{rcl}
q^{(4)}&=&!\ket{0}\otimes!\ket{0}\otimes!\ket{1}\otimes!\ket{0}\otimes!\ket{0}\\
q^{(6)}&=&!\ket{0}\otimes!\ket{0}\otimes!\ket{1}\otimes!\ket{1}\otimes!\ket{0}\\
q^{(20)}&=&!\ket{1}\otimes!\ket{0}\otimes!\ket{1}\otimes!\ket{0}\otimes!\ket{0}
\\
q^{(22)}&=&!\ket{1}\otimes!\ket{0}\otimes!\ket{1}\otimes!\ket{1}\otimes!\ket{0}
\end{array}$\\
$p_2=\sum\limits_{u \in
C(2,5,I)}{|\alpha_u|^2}=|\alpha_{4}|^2+|\alpha_{6}|^2+|\alpha_{20}|^2+|\alpha_{
22}|^2$

\noindent which represents the following quantum state:
$$\frac{1}{\sqrt{p_2}}(\alpha_4\ket{00100}+\alpha_6\ket{00110}+\alpha_{20}\ket{
10100}+\alpha_{22}\ket{10110})$$
\end{example}

\subsection{Conditional statements}
Measurement as a feature is only useful if the result of the measurement can be
used to determine the future evolution of the program. Hence a conditional
statement similar to that given in QML is needed. However, in contrast to QML's
$\mathbf{if}$ statements \cite{jon}, only base-qubits are allowed in the
condition. This is all that is required, as the \emph{if} structure is only
needed to provide a way to read the output of measurements.
Conditional statements are realised by adding the following to the syntax:
$$\cif t_1 \cthen t_2 \celse t_3$$
and the operational semantic is given by:
$$\begin{array}{c}
\inference{}{\cif !\ket{0} \cthen t_1 \celse t_2 \flecha{1} t_1 }[$\rif{0}$] \\
\\
\inference{}{\cif !\ket{1} \cthen t_1 \celse t_2 \flecha{1} t_2 }[$\rif{1}$] \\
\end{array}$$

Note that as the condition may be not be a base-qubit, it is not guaranteed that
the whole term will reduce.
\medskip

This addition is required, as without such an \emph{if} statement such as this
being added to the language,
this extension to measurements would have been equivalent
to a simple extension from unitary constants to quantum operation constants.

\subsection{Example: Teleportation algorithm}
With the rules developed so far, the teleportation algorithm can be rewritten as
shown in Fig. \ref{alg:telepM}.

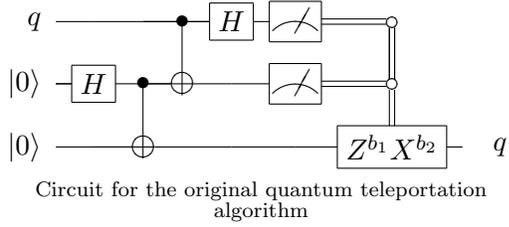
\begin{figure}[htp]
 \centering
\begin{minipage}{6.3cm}
{\small\medskip
\noindent$\teleport\ q \flecha{1}
 \begin{aligned}[t]
 	&\clet x\otimes y = \epr \cin \\
    	&\clet b_1\otimes b_2 = M_{\{1,2\}}\ \alice\ q\ x \cin \\
    	&\quad \bob\ b_1\ b_2\ y
 \end{aligned}$

where	

$\begin{array}{l}
\alice\ q\ x \flecha{1}\begin{aligned}[t]
	                           & \clet r\otimes w = \cnot\ q\otimes x \cin
\\
	                           &  ((H\ r)\otimes w)
	                          \end{aligned}\\
\bob\ b_1\ b_2\ y\flecha{1}\aplica{zed}\ b_1\ (\aplica{ex}\ b_2\ y) \\
\aplica{ex}\ b\ x\flecha{1}\cif b \cthen (X\ y) \celse y \\
\aplica{zed}\ b\ x\flecha{1}\cif b \cthen (Z\ x) \celse x \\
\epr \equiv \cnot\ ((H\ !\ket{0})\otimes !\ket{0})
\end{array}$}
\medskip
\end{minipage}
\   \
\hfill \begin{minipage}{7cm}
\begin{eqnarray*}{
\Qcircuit @C=0.55em @R=.75em {
\lstick{q} & \qw 	   & \qw & \ctrl{1} \qw & \gate{H} \qw & \meter	&
\controlo \cw \cwx[1] \\
\lstick{\ket{0}}    & \gate{H} \qw & \ctrl{1}		& \targ \qw    & \qw   
& \meter & \controlo \cw \cwx \\
\lstick{\ket{0}}    & \qw 	   & \targ \qw		& \qw          & \qw
& \qw    & \gate{Z^{b_1}X^{b_2}} \cwx & \qw & \rstick{q}
}}
\end{eqnarray*}

{\begin{center}
\small Circuit for the original quantum teleportation algorithm                 
                                              \end{center}}
\end{minipage}
\caption{Teleportation algorithm in extended $\lambda_q$} \label{alg:telepM}
\end{figure}

\section{Confluence}
When defining a language, a grammar must also be provided (how to construct
terms), and a semantics (how these terms compute). The semantics can be
denotational (terms are mapped to elements of a semantic domain, each
corresponding to what is computed by the term) or operational (terms are mapped
into other terms, with each transition corresponding to a computational step).
Clearly it must be proved that the semantics provided is unambiguous and
consistent. For example, the semantics will usually induce an equational theory
upon terms (via equality in the semantics domain or by equating two terms if one
reduces to the other), and it is important that this theory should not equate
all terms.

In $\lambda_q$ a consistent equational theory is given. However, adding
measurement does not correspond to a simple system for equational reasoning. It
is not possible to proceed by replacing terms by equal terms according to any
equational theory, since measurement is a probabilistic operation, and each
reduction instance could produce different terms that are impossible to
reconcile in the system. In the presence of an operational semantics, a usual
method of proving the consistency result is to provide a proof of confluence. 
This property states that the order in which the transition rules are applied
does not matter to the end result, thus removing any ambiguity. In this section
it is shown how such a study of confluence can still be carried through, even in
the presence of probabilities. As $\lambda_q$ provides a fixed reduction strategy, proving confluence in the original language is trivial, because there is only one possible reduction at each step. However, this is not the case in the presence of measurement, where proving confluence is non-trivial.

\subsection{Definitions and lemmas}
Whilst the above-mentioned probabilistic reductions are an elegant and concise
way to present the operational semantics, the study of confluence is not
immediate in this setting. For confluence, it is necessary to prove that if any
term $t$ can reduce to $u$ and to $v$, then there exists a $w$ such that
$u\rightarrow w \wedge v\rightarrow w$. However, in a probabilistic calculus it
could be that $t\flecha{p}u$ and $t\flecha{q}v$, where $p$ and $q$ represent the
probability of the respective reduction occurring, and there is no $w$ that both
$u$ and $v$ could reduce to. For example, given $M_{\{1\}}$, a measurement
operator in the computational basis, it follows that $M_{\{1\}}\
(\alpha\ket{0}+\beta\ket{1})\flecha{|\alpha|^2}\ket{0}$ and $M_{\{1\}}\
(\alpha\ket{0}+\beta\ket{1})\flecha{|\beta|^2}\ket{1}$. However, there is no $w$
such that $\ket{0}\flecha{p}w$ and $\ket{1}\flecha{q}w$.

A na\"{i}ve way to deal with this would be to assume that if there is some
normal form that can be reached with a certain probability, then by following
any path it must to be possible to reach the same normal form with the same
probability. However, this definition is not rigorous, and not applicable to
terms without a normal form. Hence, it does not allow the development of a
formal proof of confluence.

Probabilistic transitions need to be abstracted out in order to allow only one
possible normal form for each term, and to deal with terms without normal form.
With this aim, the following definition gives a notion of confluence for
probabilistic calculi:

\begin{definition}\label{def:termensable}
 A term ensemble $\te{t_i}{\alpha_i}$ is defined as a collection of terms $t_i$,
each with an associated probability $\alpha_i$, such that
$\sum\limits_{i}{\alpha_i}=1$.
\end{definition}

Note that given a term $t$, it may be considered as a term ensemble $\te{t}{1}$.

\begin{example}
Consider the term ensemble
$\{\langle t_1,\, \frac{1}{2}\rangle,\, \langle t_1,\, \frac{1}{4}\rangle,\,
\langle t_2,\,\frac{1}{4}\rangle\}$,
where the term $t_1$ appears twice. 
By summing the probabilities of any equivalent terms, this ensemble can be
identified with the more compact ensemble
$\{\langle t_1,\,\frac{3}{4}\rangle,\, \langle t_2,\, \frac{1}{4}\rangle\}$.
\end{example}

\begin{remark}
 Throughout this paper the symbol $=$ will be used for both
$\alpha$-equivalences and equalities. When referring to a set, \emph{i.e.} where
each element appears once, it is considered to be modulo $\alpha$-equivalence.
\end{remark}

The appropriate steps such that $\{\langle t,\, \alpha\rangle,\,\langle t,\,
\gamma\rangle\}$ is identified with $\te{t}{\alpha+\gamma}$ need to be taken.
Definition \ref{def:simpensemble} formalises this equivalence:

\begin{definition}\label{def:simpensemble}
Let $\f{first}$ be a function that takes a term ensemble and returns a set
defined by
$$\f{first}(\te{t_i}{\alpha_i}) = \{t_i\}$$
As the co-domain is a set, it allows only one instance of each element.

Let $\f{sumprob}$ be a function that takes a term and a term ensemble and
returns the sum of the probabilities associated to each instance of the term in
the ensemble:
$$\f{sumprob}(s,\te{t_i}{\alpha_i}) = \sum\limits_{j\in\{i|t_i = s\}} \alpha_j$$

Finally, let $\f{min}$ be a function that takes a term ensemble and returns a
term ensemble defined by 
$$\f{min}(\tau) = \bigcup\limits_{t\, \in\, \f{first}\ \tau}
\te{t}{\f{sumprob}(t, \tau)}$$

A term ensemble $\omega_1$ is thus said to be \emph{equivalent} to a term
ensemble $\omega_2$, $\omega_1\equiv\omega_2$, \emph{iff} $\f{min}(\omega_1) =
\f{min}(\omega_2).$
\end{definition}

Note that the definition of $\f{min}$ is correct, as $\sum\limits_{t\, \in\,
\f{first}\ \tau}\f{sumprob}(t,\tau)$ trivially sums to $1$.
\medskip

A deterministic transition rule between term ensembles can also be defined:
\begin{definition}\label{def:transition}
 If $X$ is a probabilistic rewrite system over terms, let $Det(X)$ be the
deterministic rewrite system over term ensembles written $\tte{X}$ and defined
as
$$\te{t_i}{\alpha_i}\tte{X}\te{t'_{i_j}}{\alpha_{i}\gamma_{i_j}} \textrm{
\emph{iff}, for each }i,\
t_i\tps{X}{\gamma_{i_j}}t'_{i_j}\wedge\sum\limits_{j}\gamma_{i_j}=1.$$
where all the reductions between single terms are produced by
following any rule in $X$, or none.
\end{definition}

\begin{lemma}
 Given a probabilistic rewrite system $P$, then $Det(P)$ preserves ensembles.
\end{lemma}
\begin{proof}
Let $\te{t_i}{\alpha_i}$ and $\te{t'_{i_j}}{\alpha_i\gamma_{i_j}}$ be term
ensembles such that
$\te{t_i}{\alpha_i}\tte{P}\te{t'_{i_j}}{\alpha_i\gamma_{i_j}}$. Then, by
definition \ref{def:transition}, $\forall i\ \sum\limits_{j}\gamma_{i_j}=1.$

Hence, $\sum\limits_{i,j}\alpha_i\gamma_{i_j} = \sum\limits_i \alpha_i
\sum\limits_j \gamma_{i_j} = \sum\limits_i \alpha_i = 1$.
\end{proof}
\bigskip

Using these concepts, (strong) confluence for a probabilistic rewrite system can
be expressed as show in 
definition \ref{def:confluence}.
\begin{definition}\label{def:confluence}
Let $R$ be a probabilistic rewrite system. $R$ is said to be confluent if, for
each term ensemble $\tau$ such that $\tau\tte{R}_*\mu \wedge \tau\tte{R}_*\nu$,
there exist equivalent term ensembles $\omega_1$ and $\omega_2$ such that
$\mu\tte{R}_*\omega_1 \wedge \nu\tte{R}_*\omega_2$. $R$ is said to be
\emph{strongly} confluent if, for each term ensemble $\tau$, such that
$\tau\tte{R}\mu \wedge \tau\tte{R}\nu$, there exist equivalent term ensembles
$\omega_1$ and $\omega_2$ such that $\mu\tte{R}\omega_1 \wedge
\nu\tte{R}\omega_2$.
\end{definition}
Note that strong confluence of $R$ implies the confluence of $R$, and also that
the confluence of $R$ implies the
strong confluence of $R^*$.
It is possible to extend the Hindley-Rosen lemma \cite{HRlemma}\cite{HRlemma2}
to these notions of confluence, as follows:
\begin{proposition}\label{prop:confluenceunion}
Let $R$ and $U$ be strongly confluent probabilistic rewrite systems. If $R$ and
$U$ strongly commute, that is if for each term ensemble $\tau$ such that
$\tau\tte{R}\mu\ \wedge\ \tau\tte{U}\nu$, there exist equivalent term ensembles
$\omega_1$ and $\omega_2$ such that $\mu\tte{U}\omega_1\ \wedge\
\nu\tte{R}\omega_2$, therefore $R\cup U$ is strongly confluent.
\end{proposition}
\medskip

Theorem \ref{pablo} allows the remaining proofs to be simplified, by showing
that it is enough to prove strong confluence (commutation) for a single-term
term ensemble.

\begin{theorem}\label{pablo}
Let $S$ and $T$ be probabilistic rewrite systems such that:
$$\forall t\ \left. \begin{array}{l}
\te{t}{1}\tte{S}\mu_1\\
\te{t}{1}\tte{T}\nu_1
\end{array}\right\rbrace \Rightarrow \exists\ \omega_1\equiv\omega_2\textrm{
s.t. }\left\lbrace \begin{array}{l}
\mu_1\tte{T}\omega_1\\
\nu_1\tte{S}\omega_2\\
\end{array}\right.$$

Then $\forall \tau$, $\mu$ and $\nu$ such that $\tau\tte{S}\mu$ and
$\tau\tte{T}\nu$, there exist equivalent $\omega_1$ and $\omega_2$ such that
$\mu\tte{T}\omega_1$ and $\nu\tte{S}\omega_2$.
\end{theorem}
\begin{proof}
Let $\tau=\te{t_i}{\alpha_i}$, $\mu=\te{u_{i_j}}{\alpha_i\delta_{i_j}}$ and
$\nu=\te{v_{i_k}}{\alpha_i\varphi_{i_k}}$ such that $\tau\tte{S}\mu$ and
$\tau\tte{T}\nu$, \emph{i.e.} for each $i$:
\begin{equation}\label{thm:pablo-transition}
t_i\tps{S}{\delta_{i_j}}u_{i_j}\ \wedge\ \sum\limits_j\delta_{i_j}=1 \textrm{
and }
t_i\tps{T}{\varphi_{i_k}}v_{i_k}\ \wedge\ \sum\limits_k\varphi_{i_k}=1
\end{equation}

Consider the single term term-ensembles $\tau_i=\te{t_i}{1}$, and the term
ensembles $\mu_i=\te{u_{i_j}}{\delta_{i_j}}$ and
$\nu=\te{v_{i_k}}{\varphi_{i_k}}$. By equation (\ref{thm:pablo-transition}), for
each $i$, $\tau_i\tte{S}\mu_i$ and $\tau_i\tte{T}\nu_i$. By our hypothesis, for
each $i$ there exist equivalent term ensembles
$\omega_1^i=\te{w_{1j_l}^i}{\delta_{i_j}\sigma_{j_l}^i}$ and
$\omega_2^i=\te{w_{2k_s}^i}{\varphi_{i_k}\gamma_{k_s}^i}$ such that
$\mu_i\tte{T}\omega_1^i$ and $\nu_i\tte{S}\omega_2^i$.

By taking $\omega_1=\te{w_{1j_l}^i}{\alpha_i\delta_{i_j}\sigma_{j_l}^i}$ and
$\omega_2=\te{w_{2k_s}^i}{\alpha_i\varphi_{i_k}\gamma_{k_s}^i}$, it follows that
$\mu\tte{T}\omega_1$ and $\nu\tte{S}\omega_2$. As $\forall i$
$\omega_1^i\equiv\omega_2^i$, it is trivially the case that
$\omega_1\equiv\omega_2$.
\end{proof}
\bigskip

Lemma \ref{lemacorto} guarantees that equivalence between term ensembles is a
congruence by adding identical context to each term in both of the ensembles:

\begin{lemma}\label{lemacorto}
 Given two equivalent term ensembles $\omega_1=\te{t_i}{\alpha_i}$ and
$\omega_2=\te{s_j}{\gamma_j}$ and any context $C$, the term ensembles
$\tau_1=\te{C[t_i/x]}{\alpha_i}$ and $\tau_2=\te{C[s_j/x]}{\gamma_j}$ are also
equivalent.
\end{lemma}
\begin{proof}
 $\omega_1\equiv\omega_2\Rightarrow\f{min}(\omega_1)=\f{min}(\omega_2)$, defined
as equal to $\te{w_k}{\delta_k}$, then

\begin{eqnarray*}
\f{min}(\tau_1) & = & \bigcup\limits_{t\,\in\,\f{first}(\tau_1)}
\te{t}{\f{sumprob}(t,\tau_1)} \\
 &=& \f{min}\left(\bigcup\limits_{t\,\in\,\f{first}(\omega_1)}
\te{C[t/x]}{\f{sumprob}(t,\omega_1)}\right) \\ 
 &=& \f{min}\left(\te{C[w_k/x]}{\delta_k}\right) \\
 &=& \f{min}\left(\bigcup\limits_{t\,\in\,\f{first}(\omega_2)}
\te{C[t/x]}{\f{sumprob}(t,\omega_2)}\right) \\
 &=& \bigcup\limits_{t\,\in\,\f{first}(\tau_2)} \te{t}{\f{sumprob}(t,\tau_2)} \\
 &=& \f{min}(\tau_2)
\end{eqnarray*}
and hence $\tau_1 \equiv \tau_2$.
\end{proof}

\subsection{Strong confluence for \texorpdfstring{$\{(\meas), \rif{0},
\rif{1}\}$}{M, IF-0, IF-1}}
The strong confluence of the added rules is formally expressed and proved by
theorem \ref{omegas}.

\begin{theorem}\label{omegas}
The probabilistic reduction rules system $T=\{(\meas), \rif{0}, \rif{1}\}$ is
strongly confluent.
\end{theorem}
\begin{proof}
Given term ensembles $\tau=\te{t}{1}$, $\mu$ and $\nu$, where $\mu\neq\nu$, and
such that $\tau\tte{T}\mu$ and $\tau\tte{T}\nu$, then by proving there exist
equivalent term ensembles $\omega_1$ and $\omega_2$ such that
$\mu\tte{T}\omega_1$ and $\nu\tte{T}\omega_2$, theorem \ref{pablo} shows that
this system is strongly confluent.

 This result is proved here using structural induction over $t$.
\begin{enumerate}
 \item $t=x\ |\ c_U\ |\ q\ |\ M_I\ |\ !t'\ \Rightarrow\ \nexists\ \mu\neq\nu$.
Note that there is no rule in $T$ that can reduce $t$ in this case, and hence
only $\id$ is applicable, producing $\mu=\tau$. Therefore there cannot exist any
$\nu\neq\mu$.
 \item\label{trivial} $\nu=\tau$. Hence $\omega_1=\omega_2=\mu$.
 \item\label{abstraccion} $t=\lambda x.t'$.

\noindent Let $\mu=\te{\lambda x.u_i}{\alpha_i}$ where $(t'\tps{T}{\alpha_i}
u_i)$, 
with $\sum\limits_{i}{\alpha_i}=1$, and let $\nu=\te{\lambda x.v_j}{\gamma_j}$
where $(t'\tps{T}{\gamma_j} v_j)$ with $\sum\limits_{j}{\gamma_j}=1$.

\noindent By induction, there exist equivalent term ensembles
$\omega'_1=\te{w^1_{i_k}}{\alpha_i\beta_{i_k}}$ and
$\omega'_2=\te{w^2_{j_h}}{\gamma_j\sigma_{j_h}}$ such that
$u_i\tps{T}{\beta_{i_k}}w^1_{i_k}$ and $v_j\tps{T}{\sigma_{j_h}}w^2_{j_h}$.

\noindent Hence $\omega_1 = \te{\lambda x.w^1_{i_k}}{\alpha_i\beta_{i_k}}$ and
$\omega_2 = \te{\lambda x.w^2_{j_h}}{\gamma_j\sigma_{j_h}}$ can be taken, which
are equivalent by lemma \ref{lemacorto}.
 
 \item $t=\lambda !x.t'$, analogous to case (\ref{abstraccion}).

  \item\label{aplic} $t=(t_1\ t_2)$. Consider the following cases:
  \begin{enumerate}
     	\item\label{largo} Let $\mu = \te{(t_1\ u_i)}{\alpha_i}$ where
$t_2\tps{T}{\alpha_i}u_i$, with $\sum\limits_{i}{\alpha_i}=1$, and\\
        let $\nu = \te{(t_1\ v_j)}{\gamma_j}$ where $t_2\tps{T}{\gamma_j} v_j$,
with $\sum\limits_{j}{\gamma_j}=1$.\\
     This case is analogous to case (\ref{abstraccion}), as lemma
\ref{lemacorto} is applicable.

	\item Let $\mu = \te{(u_i\ t_2)}{\alpha_i}$ and $\nu = \te{(v_j\
t_2)}{\gamma_j}$. This follows case (\ref{largo}).

	\item\label{mixed} Let $\mu = \te{(u_i\ t_2)}{\alpha_i}$ and $\nu =
\te{(t_1\ v_j)}{\gamma_j}$, then take $\omega_1=\omega_2=\te{(u_i\
v_j)}{\alpha_i\gamma_j}$

	\item Let $t = (M_I\ q)$, $\mu = \te{q_i}{\alpha_i}$ where $(M_I\
q)\tps{T}{\alpha_j} q_j$, with $\sum\limits_{i}{\alpha_i}=1$. This follows case
(\ref{trivial}). %the only solution for $\nu\neq\mu$ is $\nu=\tau$. Hence,
$\omega_1=\omega_2=\mu$.
  \end{enumerate}
  \item $t=\cif t_1 \cthen t_2 \celse t_3$. Consider the following cases:
	\begin{enumerate}
	\item Let $\mu=\te{\cif t_1\cthen t_2\celse u_i}{\alpha_i}$ where
$t_3\tps{T}{\alpha_i}u_i$, with $\sum\limits_{i}{\alpha_i}=1$ and\\
	let $\nu=\te{\cif t_1\cthen v_j\celse t_3}{\gamma_j}$ where
$t_2\tps{T}{\gamma_j}v_j$, with $\sum\limits_{j}{\gamma_j}=1$. \\
        This is analogous to (\ref{aplic}.\ref{mixed}). In fact, any combination
that implies that $\mu$ and $\nu$ are obtained by the reduction of $t_1$, $t_2$
or $t_3$, is analogous to one of the subcases of case (\ref{aplic}).
	\item Let $\mu=\te{t_2}{1}$, and let $\nu=\te{\cif t_1\cthen t_2\celse
v_j}{\gamma_j}$ where $t_3\tps{T}{\gamma_j}v_j$, with
$\sum\limits_{j}{\gamma_j}=1$. Then take $\omega_1=\omega_2=\mu$. (Analogous if
$t_2\tps{T}{\gamma_j}v_j$ and $\mu=\te{t_3}{1}$).
	\item Let $\mu=\te{t_2}{1}$, and let $\nu=\te{\cif t_1\cthen v_j\celse
t_3}{\gamma_j}$ where $t_2\tps{T}{\gamma_j}v_j$, with
$\sum\limits_{j}{\gamma_j}=1$. Then take $\omega_1=\omega_2=\te{v_j}{\gamma_j}$.
(Analogous for $t_3$).
	\end{enumerate}
\end{enumerate}
\end{proof}

\subsection{Preserving confluence} 

Before formalising the confluence for the whole calculus, some key examples are
presented:
\begin{itemize}
\item \emph{Cloning} arguments:
$(\lambda x.(x\ x))\ (M_{\{1\}}\
(\frac{1}{\sqrt{2}}!\ket{0}+\frac{1}{\sqrt{2}}!\ket{1}))$

The problem here is that if copying a measurement is allowed, this
may give different results for each measurement. However, by measuring first and
then applying
the abstraction, both measurements are the same.
In $\lambda_q$, these kinds of terms are disallowed by the well-formedness
rules \cite{qlambda}; a linear argument can appear only once in the body of a
function.

\item \emph{Copying} arguments:
$(\lambda!x.(x\ x))\ (M_{\{1\}}\
(\frac{1}{\sqrt{2}}!\ket{0}+\frac{1}{\sqrt{2}}!\ket{1}))$

When the argument is linear, there is no rule in the operational semantics of
$\lambda_q$
that allows the application of a non-linear abstraction to a linear term. Hence,
$M_{\{1\}}$ must apply first,
producing a non-linear output (either $!\ket{0}$ or $!\ket{1}$).

\item \emph{Promoting} arguments:
$(\lambda!x.(x\ x))\ !(M_{\{1\}}\
(\frac{1}{\sqrt{2}}!\ket{0}+\frac{1}{\sqrt{2}}!\ket{1}))$

In this case copying the measurement operation twice is allowed, and this is the
only
applicable reduction strategy because $!t$ terms are values in $\lambda_q$.
\end{itemize}
\bigskip

In light of the above statements, a formal proof of confluence for the entire
system is required.

Lemma \ref{sincontexto} ensures that, under some hypotheses, measurement is
independent of context:

\begin{lemma}\label{sincontexto}
 Let $x$ be a variable and let $t$ be a linear term with only one linear
instance of $x$. If $m\tps{(M)}{p}v$, then $t[m/x]\tps{(M)}{p}t[v/x]$.
\end{lemma}
\begin{proof}
Structural induction over $t$
\begin{enumerate}
 \item Let $t$ be such that $x\notin \mathcal{V}(t)$ $\Rightarrow\
t[m/x]=t=t[v/x]$.
 \item Let $t=x$. $x[m/x] = m \tps{(M)}{p}v = x[v/x]$.
 \item\label{lambda} Let $t=\lambda y.t'$. By induction $\lambda y.t'[m/x]
\tps{(M)}{p}\lambda y.t'[v/x]$.
 \item Let $t=\lambda !y.t'$. Analogous to case (\ref{lambda}).
 \item\label{app} Let $t=(t_1\ t_2)$, with $x \in \mathcal{V}(t_1)$. Then $(t_1\
t_2)[m/x]$ = $(t_1[m/x]\ t_2)$ and by induction, $(t_1[m/x]\
t_2)\tps{(M)}{p}(t_1[v/x]\ t_2)$, which is equal to $(t_1\ t_2)[v/x]$.
 \item Let $t=(t_1\ t_2)$, with $x \in \mathcal{V}(t_2)$. Analogous to case
(\ref{app}).
 \item Let $t=\cif t_1\cthen t_2\celse t_3$. Analogous to case (\ref{app}).
\end{enumerate}
\end{proof}

Next, it is proved that the original reduction rules system from $\lambda_q$ and
the new rules for measurements strongly commute. This is suggestive of the
confluence of the whole system 

\begin{theorem}\label{thm2}
 The probabilistic reduction rules systems $S=\{(APP_1), (APP_2),$ $(\beta),
(!\beta_1), (!\beta_2), (U)\}$ and $T=\{(\meas), \rif{0}, \rif{1}\}$ strongly
commute.
\end{theorem}
\begin{proof}
If it is proved that given term ensembles $\tau=\te{t}{1}$, $\mu$ and $\nu$,
$\mu\neq\nu$, such that $\tau\tte{S}\mu$ and $\tau\tte{T}\nu$, then this implies
that there exist equivalent term ensembles $\omega_1$ and $\omega_2$ such that
$\mu\tte{T}\omega_1$ and $\nu\tte{S}\omega_2$, then $S$ and $T$ verify the
hypotheses for theorem \ref{pablo}, which proves strong commutation between
them.

 This result is proved here using structural induction over $t$.
\begin{enumerate}
 \item $t=x\ |\ c_U\ |\ q\ |\ M_I\ |\ !t\ \Rightarrow\ \nexists\ \mu\neq\nu$.
Note that there is no rule in $T$ nor $S$ that can reduce $t$ in this case,
hence only $\id$ is applicable, producing $\mu=\tau$. Therefore there cannot
exists any $\nu\neq\mu$.
 \item\label{thm2:trivial} $\nu=\tau$. Hence $\omega_1=\omega_2=\mu$. (Analogous
for $\mu=\tau$).
 \item\label{thm2:lambda} $t=\lambda x.t'$, $\mu=\te{\lambda x.u}{1}$ and
$\nu=\te{\lambda x.v_j}{\gamma_j}$ such that $\tau\tte{S}\mu$ and
$\tau\tte{T}\nu$. By induction, there exist equivalent
$\omega'_1=\te{w^1_{s}}{\delta_s}$ and $\omega'_2=\te{w^2_{j}}{\sigma_j}$ such
that $\te{u}{1}\tte{T}\omega'_1$ and $\te{v_j}{\gamma_j}\tte{S}\omega'_2$.
Then take $\omega_1=\te{\lambda x.w^1_{s}}{\delta_s}$ and $\omega_2=\te{\lambda
x.w^2_{j}}{\sigma_j}$ which are equivalent by lemma \ref{lemacorto}.
 \item $t=\lambda!x.t'$. Analogous to case (\ref{thm2:lambda}).
 \item\label{thm2:app} $t=(t_1\ t_2)$. Consider the following cases:
 \begin{enumerate}
  \item\label{thm2:aplica} $\mu=\te{(t_1\ u)}{1}$ and $\nu=\te{(t_1\
v_j)}{\gamma_j}$. Analogous to case (\ref{thm2:lambda}); note that lemma
\ref{lemacorto} also holds in this case.
  \item\label{thm2:aplica2} $\mu=\te{(u\ t_2)}{1}$ and $\nu=\te{(v_j\
t_2)}{\gamma_j}$. Analogous to 
subcase (\ref{thm2:aplica}).
  \item\label{thm2:mixed} $\mu=\te{(u\ t_2)}{1}$ and $\nu=\te{(t_1\
v_j)}{\gamma_j}$.
	Take $\omega_1=\omega_2=\te{(u\, v_j)}{\gamma_j}$ (Similarly if
$t_2\flecha{1}u$ and $t_1\flecha{\gamma_j}v_j$).
  \item\label{aplic:gate} $t=(c_U\ q)$ and $\mu=\te{q'}{1}$, This follows case
(\ref{thm2:trivial}). Note that if instead of $t_2=q$ an expression like
$t_2=(M_I\ q)$ is given, it is subcase (\ref{thm2:aplica2}) which applies, where
$u=t_1=c_U$.
  \item $t=(M_I\ q)$ and $\mu=\tau$. This follows the analogous to case
(\ref{thm2:trivial}).
  \item $t_1=\lambda x.t'$, $\mu=\te{t'[t_2/x]}{1}$, $\nu=\te{(\lambda x.t'\
v_j)}{\gamma_j}$. By lemma \ref{sincontexto},
$\omega_1=\omega_2=\te{t'[v_j/x]}{\gamma_j}$ can be taken. Note that if
$t_1=\lambda!x.t'$, with the same $\mu$, then $t_2$ must be non-linear due to
the well-formedness rules and hence in this situation it is the subcase
(\ref{aplic:gate}).
  \item $t_1=\cif t'_1\cthen t'_2\celse t'_3$. Then $\mu$ has to be obtained by
the reduction of $t'_1, t'_2, t'_3$ or $t_2$, hence, it is analogous to previous
cases. Note that if, for instance, $\nu=\te{(t'_2\ t_2)}{1}$ and suppose that
$\mu$ is obtained by the reduction of $t'_3$ (it cannot be the application of
the$\cif$statement to $t_2$ because there is not any rule that performs such a
reduction) then $\omega_1=\omega_2=\nu$.
 \end{enumerate}
 \item $t=\cif t_1\cthen t_2\celse t_3$. Consider the following cases:
 \begin{enumerate}
  \item Let $\mu=\te{\cif t_1\cthen t_2\celse u}{1}$ where $t_3\tps{S}{1}u$ and
let $\nu=\te{\cif t_1\cthen v_j\celse t_3}{\gamma_j}$ where
$t_2\tps{T}{\gamma_j}v_j$ and $\sum\limits_{j}{\gamma_j}=1$. Analogous to
(\ref{thm2:app}.\ref{thm2:mixed}). In fact, any combinations that implies that
$\mu$ and $\nu$ are obtained by reduction of $t_1$, $t_2$, or $t_3$, is
analogous to one of the subcases of case (\ref{thm2:app}).
  \item Let $\mu=\te{\cif t_1\cthen t_2\celse u}{1}$ where $t_3\tps{S}{1}u$ and
$\nu=\te{t_2}{1}$, then take $\omega_1=\omega_2=\nu$. Analogous if
$t_2\tps{S}{1}u$ and $\nu=\te{t_3}{1}$.
  \item Let $\mu=\te{\cif t_1\cthen u\celse t_3}{1}$ where $t_2\tps{S}{1}u$ and
$\nu=\te{t_2}{1}$, then take $\omega_1=\omega_2=\te{u}{1}$. Similarly if
$t_3\tps{S}{1}u$ and $\nu=\te{t_3}{1}$.
 \end{enumerate}
\end{enumerate}
\end{proof}

It has been shown that $T$ and $S$ strongly commute, and hence $T^*$ and $S^*$
strongly commute. Moreover, $T$ is confluent, and hence $T^*$ is strongly
confluent.

Now, supposing $S$ is confluent, it follows that $S^*$ is strongly confluent.
Proposition \ref{prop:confluenceunion} entails that $S^*\cup T^*$ is strongly
confluent, and therefore that $S\cup T$ is confluent. Therefore, the extension
of van Tonder's calculus presented here preserves confluence.

\section{Conclusions}
This paper extends the quantum lambda calculus $\lambda_q$, defined by van
Tonder, with a family of measurement operations $M_I$, which measure the qubits
indicated by the set $I$, and an \emph{if} structure which allows reading of the
output of these measurements. By defining the notion of ensembles of terms, and
extending the rewrite system to a deterministic system between term ensembles, a
proof of confluence for this extended calculus is presented. The extended
calculus is therefore confluent, and retains the simplicity of van Tonder's
original calculus.

The proof of confluence follows a method which can be applied to other calculi
that make use of probabilistic transition rules. For example, this method could
be applied to both {\em Lineal} and to {\em QML}, and this is the subject of
ongoing research.

The addition of a measurement operation to $\lambda_q$, which preserves
confluence, is a significant development. This allows a more natural expression
of quantum algorithms that intrinsically make use of measurement, such as
quantum teleportation, superdense coding, and quantum search algorithms.
Moreover, having an operational semantic for measurements gives a way for
understanding the behaviour of this quantum procedure, and this is a possible
topic for future work.

\begin{ack}
 A. D\'{\i}az-Caro would like to thank Pablo E. Mart\'{\i}nez L\'{o}pez for useful comments and helpful suggestions on an early draft of this paper, and the CAPP (QCG) group at the Laboratoire d'Informatique de Grenoble for their hospitality. The authors would also like to thank Simon Perdrix for fruitful discussions.
\end{ack}

\end{document}